\documentclass[aps,pra,twocolumn,superscriptaddress]{revtex4}
\usepackage{amsfonts,epsfig}
\usepackage{amsmath}
\usepackage{amssymb}
\usepackage{amsfonts}
\usepackage{amsthm}
\usepackage{mathrsfs}
\usepackage{color,verbatim,graphics}
\usepackage{url}
\usepackage{multirow}

\newcommand{\ket}[1]{\left |  #1 \right \rangle}
\newcommand{\bra}[1]{\left \langle #1  \right |}

\newcommand{\norm}[1]{\left|\left|#1\right|\right|}

\newcommand{\etal}{{\it{et al.}}}

\newcommand{\be}{\begin{equation}}
\newcommand{\ee}{\end{equation}}
\newcommand{\ba}{\begin{eqnarray}}
\newcommand{\ea}{\end{eqnarray}}
\newcommand{\ban}{\begin{eqnarray*}}
\newcommand{\ean}{\end{eqnarray*}}

\newtheorem*{theorem*}{Theorem}

\newcommand{\bracket}[3]{\langle#1|#2|#3\rangle}

\newcommand{\BB}{\mathcal{B}}

\newcommand{\II}{\mathbb{I}}

\newcommand{\pa}[2]{\Pi^{A_#1}_{#2}}
\newcommand{\pb}[2]{\Pi^{B_#1}_{#2}}

\begin{document}

\title{Robust Self Testing of Unknown Quantum Systems into Any Entangled Two-Qubit States}

\author{Tzyh Haur Yang}
\affiliation{Centre for Quantum Technologies, National University of Singapore, 3 Science drive 2, Singapore 117543}
\author{Miguel Navascu\'es}
\affiliation{School of Physics, University of Bristol, Tyndall Avenue, Bristol BS8 1TL (U.K.)}

\begin{abstract}
Self testing is a device independent approach to estimate the state and measurement operators, without the need to assume the dimension of our quantum system. In this paper, we show that one can self test black boxes into any pure entangled two-qubit state, by performing simple Bell type experiments. The approach makes use of only one family of two-inputs/two-outputs Bell inequalities. Furthermore, we outline the sufficient conditions for one to self test any dimensional bipartite entangled state. All these methods are robust to small but inevitable experimental errors.
\end{abstract}

\maketitle

%%%%%%%%%%%%%%%%%%%%%%%%%%%%%%%%%%%%%%%%%%%%%%
\section{Introduction}

In the device independent approach to Quantum Information Processing (QIP), no assumptions are made about the states under observation, the experimental measurement devices, or even the dimensionality of the Hilbert spaces where such elements are defined. Rather, the security or reliability of a given quantum communication protocol is established from the statistics generated when several space-like separated parties measure a shared quantum state. Initially motivated by the need to take into account experimental imperfections in Quantum Key Distribution (QKD) \cite{norbert}, device-independent QIP has experienced a rapid growth in the last few years. Protocols such as QKD \cite{qkd1,qkd2,qkd3,qkd4}, randomness generation \cite{random1,random2}, entanglement swapping \cite{rafael} and teleportation \cite{teleport}, originally implemented with trusted measurement devices, have been successfully translated to the device-independent realm.

In 2004, the authors of \cite{yao} proposed a device independent scheme to certify the presence of a quantum state and the structure of a set of experimental measurement operators. This inspired further works by \cite{Bardyn2009} and \cite{McKague2012}, who respectively considered (and solved) the problem of estimating the fidelity(norm difference) between a physical bipartite state and the two-qubit maximally entangled state given the violation of the Clauser-Horn-Shimony-Holt (CHSH) Bell inequality \cite{chsh} and only Local Operations and Classical Communication (LOCC) are allowed. In \cite{McKague2012}, it was also introduced a general framework for quantum self testing, whereby two or more space-like separated parties are said to share a given state $\ket{\Psi}$ iff there exist local transformations which allow them to distill it. Independently, a general scheme for robust self testing in the context of binary nonlocal games, similar to CHSH, was presented in \cite{michigan}. Most recently, the authors in \cite{reichardt} showed that the robustness of such games are optimal and can be extended to cover the scenarios when adversaries try to use the memory in the apparatus to cheat.

%All the above results apply only to maximally entangled qubit states, which are arguably the most entangled and most nonlocal resource. It is tempting to ask what is the relationship between self testing and nonlocality. Furthermore, it is of fundamental importance to consider to possibility of self testing partially entangled qubit states, going a step further, whether one can self any pure entangled states. The possibility of self testing any arbitrary quantum system gives a sense of \note{closeness} to quantum mechanics.

The above results show how maximally entangled qubit states can be self tested in different nonlocality scenarios. Nothing is said, however, about the possibility of self testing more general quantum states. Can partially entangled qubit states be self tested as well? Or, going a step further, can any pure entangled state be self tested? A positive answer to these questions would pave the way to a complete device independent reformulation of quantum mechanics.

In this paper, we show that self testing is not only possible for maximally entangled qubits, but it can be done for arbitrary bipartite entangled qubit states, using only a simple family of two-input/two-output Bell inequalities. Likewise, we identify sufficient conditions to self test high dimensional entangled states and provide a Bell scenario that allows one to self test general $d$ dimensional maximally entangled states. Note that recently and independently, Rafael {\it et. al.} in \cite{rafael2} have illustrated a method to self test Hardy inequality, which involves partially entangled qubit states. Their method however, are not known to be robust and only limited to Hardy's inequality.

\section{The self testing scenario}
Picture a scenario where two distant observers, Alice and Bob, perform measurements over a shared quantum state $\ket{\psi'}$ (since we do not assume the dimension, the state can be taken to be pure while the measurements are projective \cite{nielson}). Let $\{\Pi^x_a,\Pi^y_b\}$ be Alice's and Bob's Positive Operator Valued Measure (POVM) elements, where $(x,y)$ labels the different measurement settings; and $(a,b)$, the measurement outcomes. The statistics that they observe will thus be given by $p(a,b|x,y)=\bracket{\psi'}{\Pi^x_a\otimes \Pi^y_b}{\psi'}$. The self testing problem consists in deciding if the knowledge of $p(a,b|x,y)$ allows one to deduce the structure of the quantum system $\{\ket{\psi'},\Pi^x_a,\Pi^y_b\}$.

To do this, we need the concept of isometry. Isometry is a linear map, $\Phi$ which maps from a Hilbert space, $\mathcal{H}_1$ to another Hilbert space, $\mathcal{H}_2$ that preserves inner products, that is $\Phi:\mathcal{H}_1\rightarrow \mathcal{H}_2$. Since probabilities are invariant under isometry maps, any two quantum systems related by a local isometry must be regarded as identical in this formalism. Therefore, if such correlations $p(a,b|x,y)$ allow Alice and Bob to infer the existence of a local isometry $\Phi=\Phi_A\otimes\Phi_B$, a state $\ket{\psi}$ and projection operators $M^x_a,N^y_b$ (satisfying commutation relations of the type $[M^x_a,N^y_b]=0$) such that
\begin{eqnarray}
	&&\Phi(\ket{\psi'})=\ket{junk}\otimes\ket{\psi},\nonumber\\
	&&\Phi([\Pi^x_a\otimes \Pi^y_b]\ket{\psi'} )=\ket{junk}\otimes M^x_a \otimes N^y_b\ket{\psi},
\end{eqnarray}

\noindent we will then say that Alice and Bob have self tested the system $\{\ket{\psi},M^x_a,N^y_b\}\widetilde{=}\{\ket{\psi'},\Pi^x_a,\Pi^y_b\}$. Note that the junk state $\ket{junk}$ are any physical states which will be traced out subsequently and thus not taken into consideration.

Sometimes, the full knowledge of $p(a,b|x,y)$ is not necessary for self testing. Indeed, in \cite{McKague2012}, McKague \etal\; showed that, in a Bell experiment with CHSH violation close to the Tsirelson bound \cite{Tsirelson1980}, 
\begin{align}
	\bracket{\psi'}{A_0(B_0+B_1)+A_1(B_0-B_1)}{\psi'} \geq 2\sqrt{2}-\epsilon,
\end{align}
for reasonably small $\epsilon$, there exists a local isometry $\Phi$ which transforms the state and operators to a state $\epsilon'$-close to the two-qubit singlet and corresponding projective measurements. The operators $A_i$ and $B_j$ here can be any measurement operators with two outcomes(dichotomic observable) on Alice and Bob's side respectively. More concretely, 
\be
	\|\Phi(\ket{\psi'})-\ket{junk}\otimes \ket{\psi}\|\leq \epsilon',
%\|\Phi((M^x_a)'(N^y_b)'\ket{\psi'})-\ket{junk}\otimes M^x_aN^y_b\ket{\phi^{+}}\|\leq \epsilon',
\ee
\noindent where $\epsilon'=\epsilon'(\epsilon)$ satisfies $\lim_{\epsilon\to 0}\epsilon'=0$. In this case, we say that the self testing is robust.

Using this concept and framework, we will first show how one can self test any partially entangled two-qubit state using a more general family of Bell inequalities.

\section{Self Testing of Partially Entangled Qubits} 
To understand how one can self test any two qubit partially entangled state, it is instructive to consider the following simple, yet illuminating, scenario. 

Suppose that Alice and Bob share the state $\ket{\psi}=\cos\theta\ket{00}+\sin\theta\ket{11}$, and act on it with the (hermitian) Pauli matrices $X$ and $Z$ on both sides. Then one can check that the relations
\begin{align}
	Z_A\ket{\psi}&=Z_B\ket{\psi}, \nonumber\\
	\sin\theta X_A(I+Z_B)\ket{\psi}&=\cos\theta X_B (I-Z_A)\ket{\psi}. \label{conditions}
\end{align}
\noindent One can obviously generate more identities with this state and operators, but these two will be enough.

We wonder whether conditions (\ref{conditions}) are sufficient for self testing. In other words, picture a Bell scenario where Alice and Bob, from their shared correlations, infer the existence of pairs of dichotomic local observables (of unknown dimensionality) $\{X_A',Z_A'\}$ and $\{X_B',Z_B'\}$ which, acting on their state $\ket{\psi'}$, satisfy conditions (\ref{conditions}). Is this sufficient for Alice and Bob to conclude that $\ket{\psi'}\widetilde{=}\ket{\psi}$?

Surprisingly, the answer is yes. It can be verified that the circuit described in \cite{McKague2012,McKague2010,McKague2010phd}, reproduced here in Figure \ref{fig:circuit1}, allows Alice and Bob to transform the state $\ket{\psi'}_{AB}\ket{00}_{AB}$ into $\ket{junk}\otimes\ket{\psi}$. For simplicity, we refer readers to \cite{supplementary} Section C for the details of how this circuit works.

\begin{figure}[htbp!]
	\includegraphics[width=0.44\textwidth]{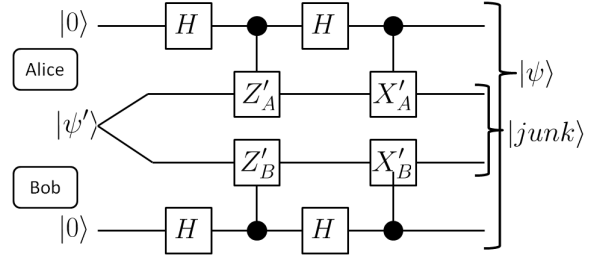}
	\caption{Local isometry, $\Phi=\Phi_A\otimes\Phi_B$, allowing Alice and Bob to self test their state. The gates $H$ are the Hadamard gate while the gates $Z_{A/B}$ and $X_{A/B}$ are the control $Z_{A/B}$ and control $X_{A/B}$ gates respectively. \label{fig:circuit1}}
\end{figure}

The problem now is to identify which correlations would allow Alice and Bob to derive relations (\ref{conditions}). Not surprisingly, all such correlations must violate a particular non-trivial Bell inequality maximally and uniquely. For our purpose, it is sufficient to consider a particular family of Bell inequalities, first studied in \cite{Acin2012}, parametrized as
\begin{align}
	\BB(\alpha) &\equiv \alpha A_0 + A_0(B_0+B_1)+A_1(B_0-B_1), \label{bell}
\end{align}
\noindent where $0\leq\alpha\leq2$. As proven in \cite{Acin2012}, the maximum quantum violation of (\ref{bell}) is given by $b(\alpha)\equiv \max_\phi\; \bracket{\phi}{\BB(\alpha)}{\phi} =\sqrt{8+2\alpha^2}$. We are ready to state our first result.

\begin{theorem*} \label{theorem1}
In any black box bipartite experiment achieving the maximum quantum violation of the Bell inequality (\ref{bell}), the corresponding quantum state is equivalent, up to local isometries, to the partially entangled state, $\cos\theta\ket{00}+\sin\theta\ket{11}$, with $\tan\theta=\sqrt{\frac{4-\alpha^2}{2\alpha^2}}$. Furthermore, this result is robust.
\end{theorem*}
\begin{proof}
	Let us first rewrite the Bell operator as $\overline{\BB}(\alpha)\equiv b(\alpha)-\BB(\alpha)$. By definition, $\overline{\BB}(\alpha)$ is positive semidefinite. It can be shown (see \cite{supplementary} Section A) that $\overline{\BB}(\alpha)$ can be expressed as $\overline{\BB}(\alpha)=\sum_\lambda P_\lambda^\dag P_\lambda$, where the $P_\lambda$s are linear functions of the operators $\mathbb{I},A_i,B_j$ and $A_iB_j$, and $i,j\in\{0,1\}$ (of course, each $P_\lambda$ depends on $\alpha$). In the event of maximum violation of the Bell operator with the state $\ket{\psi'}$, we must have $\bracket{\psi'}{\overline{\BB}(\alpha)}{\psi'}=0$, and, consequently, $P_\lambda\ket{\psi'}=0$ for all $\lambda$. By working out the exact expression for $P_\lambda$, it can be shown that such identities imply relations (\ref{conditions}) with $\tan\theta=\sqrt{\frac{4-\alpha^2}{2\alpha^2}}$ (see \cite{supplementary} Section B). Since conditions (\ref{conditions}) are sufficient for self testing, this completes the proof. 

	As for the robustness of the result, note that, if $\bra{\psi'}\overline{\BB}(\alpha)\ket{\psi'}=\epsilon^2>0$, then we have that $\|P_\lambda\ket{\psi'}\|\leq \epsilon$ for all $\lambda$. For small values of $\epsilon$, this condition implies that the quantum circuit depicted in Fig. \ref{fig:circuit1} would return a quantum state close to $\ket{\psi}$ (see \cite{supplementary} Section C).
\end{proof}

Our result covers the singlet self testing scenario considered in \cite{McKague2012} as a special case when $\alpha=0$, and extends it to any pure entangled state. A consequence of our results is that all states violating the Bell inequality (\ref{bell}) are unique, up to isometry. Complementing the findings of \cite{masanes2005}, which show that any two-setting/two-outcomes Bell inequality can be maximally violated by pure entangled qubits, our result suggests that this is also necessarily the case, at least for the family given by (\ref{bell}).

\section{High dimensional self testing}
There are two steps in this part: 1)We need to generalize the circuit in Figure \ref{fig:circuit1} for higher dimensional states; 2)We need to define an appropriate Bell scenario, possibly with higher number of inputs/outputs, whose correlation allows us to self test. 
\begin{figure}[htbp!]
	\includegraphics[width=0.4\textwidth]{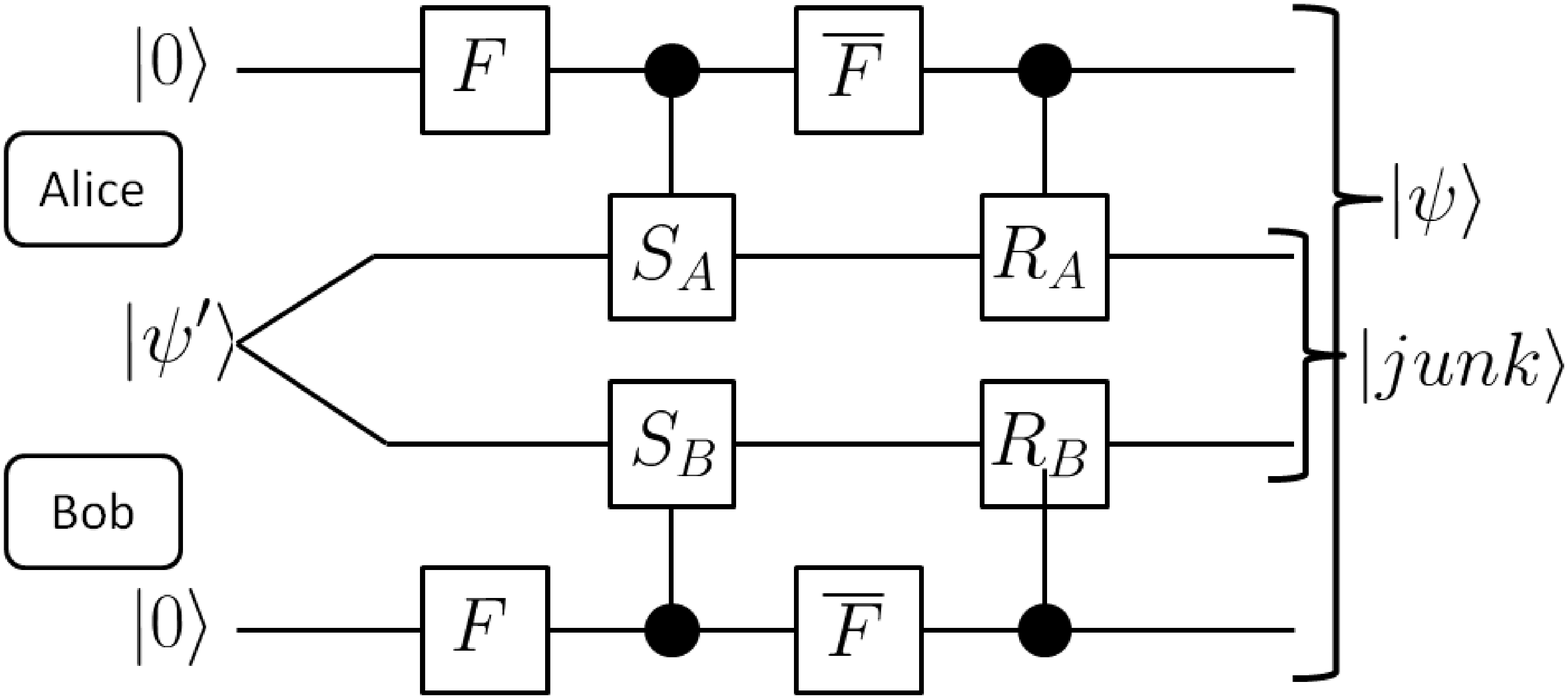}
	\caption{ Circuit to self test pure entangled state of any dimension. The gates $F$ and $\overline{F}$ here are the generalized Fourier transform gates in higher dimensions. The $S_{A/B}$ and $R_{A/B}$ gates here are the control $S_{A/B}$ and control $R_{A/B}$ gates respectively. \label{fig:circuit2}}
\end{figure}
The first step is straightforward. The generalization of the circuit in Figure \ref{fig:circuit1} is as shown in Figure \ref{fig:circuit2}. The Fourier transform gate $F$ and its inverse $\overline{F}$ are defined as $F\ket{j}=\sum_{k=0}^{d-1} \omega^{jk}\ket{k}/\sqrt{d}$ where the $d$ is the dimension of the pure state we wish to self test our black box into and $\omega$ is the $d$-th root of the unity. The action of the control phase gates $R$ and $S$ are given by $\ket{k}\ket{\psi'}\rightarrow \ket{k}X^{(k)}\ket{\psi'}$ and $\ket{k}\ket{\psi'}\rightarrow \ket{k}Z^k\ket{\psi'}$ respectively. In Supplmentary Material Section D, we show that as long as $X^{(k)}$ and $Z$ are unitary operators satisfing the following relations,
\begin{align}
	Z_{A/B} &= \sum_{i=0}^{d-1} \omega^i P_{A/B}^{(i)}, \nonumber\\
	P_A^{(i)}\ket{\psi'}&=P_B^{(i)}\ket{\psi'}, \;\;\forall i, \nonumber \\
	X_{A}^{(i)}P_{B}^{(i)}\ket{\psi'}&=\tan\delta_i (X_{B}^{(i)})^\dag P_{A}^0 \ket{\psi'}, \;\;\forall i  \label{highdd}
\end{align}
for any angles $\delta_i$, then the state $\ket{\psi'}$ can be self tested into a pure entangled states of dimension $d$. In \cite{supplementary} Section D, we also show how one can use certain nonlocal correlations between Alice and Bob $p(a,b|x,y)$ to deduce the conditions in (\ref{highdd}).

We proceed now to the second step, which is to find the correlation allowing us to deduce the existence of $X^{(k)},Z$. One possibility would be to look for a family of Bell inequalities with a higher number of measurement settings and outcomes, and figure out its Sum of Squares (SOS) representation $\sum_\lambda P^\dag_\lambda P_\lambda$. Unfortunately, finding the exact SOS decomposition of a general Bell inequality is very difficult, if not impossible. Indeed, numerical evidence suggests that certain Bell inequalities in the three settings/two outcomes scenario may not have an optimal SOS decomposition \cite{I3322}. In this paper, we will only demonstrate how one can self test high dimensional maximally entangled states, namely $\ket{\psi}=\sum_{j=0}^{d-1} \ket{ii} /\sqrt{d}$.

The technical details can be found in the \cite{supplementary} Section D.2. Here we will only sketch the intuition and simplified proof for the case when $d=4$. The Bell experiment we consider involves 3 measurements, $(A_0,A_1,A_2)$ on Alice's side and $(B_0,B_1,B_2)$ on Bob's side. First of all, consider the measurements $(A_0,A_1,B_0,B_1)$ with the correlations shown in Figure \ref{fig:correlation1}.

\begin{figure}[htbp!]
	\includegraphics[width=0.25\textwidth]{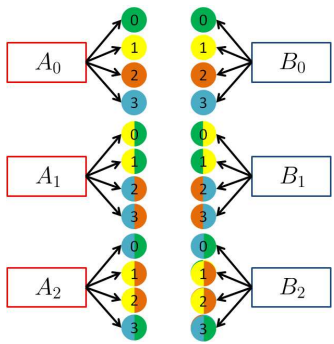}
	\caption{(Color online) The three measurements performed by Alice and Bob. All the outcomes have non zero correlations only with outcomes sharing the same color. For instance, all outcomes $\pa{0}{1}$ have zero correlation with $\pb{0}{r'}$ except $r'\in\{0,1\}$.  \label{fig:correlation1}}
\end{figure}
Due to the nature of the correlations, the first two outcomes of the four measurements $(A_0,A_1,B_0,B_1)$ in Figure \ref{fig:correlation1} can be grouped into a subspace of the total Hilbert space and be used to self test half of the maximally entangled state $\ket{\psi}$. On the other hand, the third and last outcomes can be grouped to self test the remaining half of the entangled states, as shown below
\begin{align}
	\ket{\psi}=\dfrac{1}{2} \Big( \underbrace{\ket{00} + \ket{11}}_\textrm{outcomes 0,1} + \underbrace{\ket{22}+\ket{33}}_\textrm{outcomes 2,3} \Big).
\end{align}
To complete the self testing of the state, we require the measurements $A_2$ and $B_2$ which has the correlations as shown in Figure \ref{fig:correlation1}. This time, the outcomes 1 and 2 of the measurements $(A_0,A_2,B_0,B_2)$ can be used to self test the part
\begin{align}
	\ket{\psi}=\dfrac{1}{2} \Big( \ket{00} + \underbrace{ \ket{11}+\ket{22}}_\textrm{outcomes 1,2} +\ket{33}  \Big).
\end{align}
With this, we then successfully self test the maximally entangled states with $d=4$. It is easy to see that to self test dimension $d$ maximally entangled states, we only require 3 measurements on Alice side $(A_0,A_1,A_2)$ and on Bob's side $(B_0,B_1,B_2)$. Firstly, $A_0,A_1,B_0,B_1$ are used to self test separately the subspace labelled with $(0,1), (2,3),\ldots,(d-2,d-1)$. On the other hand, the measurements $A_0,A_2,B_0,B_2$ are used to self test the subspaces labelled with $(1,2),(3,4),\ldots,(d-3,d-2)$. It is easy to see that this is then sufficient to self test the whole state.

\section{Discussion} 
There are many instances when one is interested to determine the state of the system based only on the correlations. For instance, to determine the structure of the state \cite{brunner} or simply the entanglement within the system \cite{Bardyn2009}. To have a Bell inequality which is violated maximally only by a particular state is often difficult, even semi device independently (assuming the dimension). Our results provide a simple way to designing a Bell experiment to identify uniquely the pure quantum state if 4 measurement settings are allowed.

Another important implication of our result is related to a well known result in \cite{masanes2005}, showing that it is sufficient to consider 2 qubits together with projective measurements to obtain the extremal correlations in quantum set. Here, our results show that in fact it is necessary: all (extremal) quantum systems violating the inequality in (\ref{bell}) maximally are in fact a qubit system, modulo some local isometry. An open question is then whether this is true for other type of inequality possibly with higher number of settings.

Lastly, Algebraic Quantum Field Theory shows that the concept of Hilbert space becomes redundant once operator norms are specified \cite{halvorson}. In the same spirit, our result suggests that, for experiments involving demolition measurements, even the concept of quantum states is unnecessary: the knowledge of a finite amount of probabilities may be enough to specify the system completely. 

\section{Conclusion} 
In this paper, we show explicitly how to self test any pure bipartite qubit states. Furthermore, we also show how one can self test any dimensional maximally entangled states with a scheme which is remarkably economical in terms of quantum operations: namely, the number of measurement settings is always three, and the number of measurement outcomes only grows linearly with $d$. 

\section{Acknowledgement} 
We would like to acknowledge interesting discussions with Antonio Ac\'in, Valerio Scarani and Matthew McKague. This work was supported by the National Research Foundation and the Ministry of Education, Singapore.

\appendix

\onecolumngrid

\subsection{SOS Expansion \label{append:sos}}
\label{SOS}
For clarity purposes, we reproduce the Bell inequality here:
\begin{align}
	\BB(\alpha) &\equiv \alpha A_0 + A_0(B_0+B_1)+A_1(B_0-B_1), \label{bell2}
\end{align}
where the maximum quantum violation is given by $b(\alpha)=\sqrt{8+2\alpha^2}$. 

Redefine the Bell operator as $\overline{\BB}(\alpha)\equiv b(\alpha)-\BB(\alpha)$. We would like to show that we can write $\overline{\BB}(\alpha)$ as the following sum of squares (SOS)
\begin{align}
	\overline{\BB}(\alpha) = \sum_{\lambda=1}^5 P_\lambda^\dag P_\lambda, \label{sos}
\end{align}
where $P_\lambda$ are polynomials of the form $P_\lambda=\vec{q}_\lambda \cdot \vec{V}$, where $\vec{q}_\lambda\in {\mathbb R}^9$ and $\vec{V}$ is a vector operator of the form
\begin{align}
	\vec{V}=(\mathbb{I},A_0,A_1,B_0,B_1,A_0B_0,A_0B_1,A_1B_0,A_1B_1)^T. \label{v}
\end{align}
A systematic way of obtaining numerical approximations to the SOS decomposition of a Bell operator can be found in [24,25] but we will simply provide the exact form of $\vec{q}_\lambda$ such that (\ref{sos}) holds. 

We first define 5 vectors $\vec{r}_\lambda$ as follow
\begin{align*}
	\vec{r}_1 &= \left( \begin{array}{ccccccccc} 
		0 ,&
		\dfrac{-2}{\sqrt{1+s^2}},&
		0,&
		1,&
		1,&
		0,&
		0,&
		0,&
		0\end{array}\right),  \\
		\vec{r}_2 &= \left( \begin{array}{ccccccccc} 
		\dfrac{-1}{\sqrt{1+s^2}} ,&
		\dfrac{1}{c\sqrt{1+s^2}},&
		\dfrac{s^2}{c\sqrt{1+s^2}},&
		-c^{-1},&
		0,&
		1,&
		0,&
		0,&
		0  \end{array}\right), \\
		\vec{r}_3 &= \left( \begin{array}{ccccccccc} 
		\dfrac{-1}{\sqrt{1+s^2}} ,&
		\dfrac{-1}{c\sqrt{1+s^2}},&
		\dfrac{-s^2}{c\sqrt{1+s^2}},&
		c^{-1},&
		0,&
		0,&
		1,&
		0,&
		0  \end{array}\right),  \\
		\vec{r}_4 &= \left( \begin{array}{ccccccccc} 
		\dfrac{-1}{\sqrt{1+s^2}} ,&
		\dfrac{-s^2}{c\sqrt{1+s^2}},&
		\dfrac{-1}{c\sqrt{1+s^2}},&
		c^{-1},&
		0,&
		0,&
		0,&
		1,&
		0  \end{array}\right), \\
		\vec{r}_5 &= \left( \begin{array}{ccccccccc} 
		\dfrac{1}{\sqrt{1+s^2}}, &
		-\dfrac{c+c^{-1}}{\sqrt{1+s^2}},&
		\dfrac{-1}{c\sqrt{1+s^2}},&
		c^{-1},&
		0,&
		0,&
		0,&
		0,&
		1  \end{array}\right),
\end{align*}
where $s=\sin(2\theta)$ and $c=\cos(2\theta)$. Now we can define the vectors $\vec{q}_\lambda$ as follow. 
\begin{align}
	\vec{q}_1&= \dfrac{ \gamma}{20\sqrt{2}}(\vec{r}_5-\vec{r}_4) -\dfrac{2}{5}\vec{r}_1 , \nonumber\\
	\vec{q}_2&= \left(\dfrac{3-\overline{c} }{16}\right)^{1/4}(\vec{r}_1+c\vec{r}_2-c\vec{r}_3), \nonumber\\
	\vec{q}_3 &= \dfrac{2\gamma-25c\sqrt{3-\overline{c}}}{30\sqrt{2}}\vec{r}_1+\dfrac{3}{10}(\vec{r}_5-\vec{r}_4),\nonumber \\
	\vec{q}_4 &=  \dfrac{35}{100}(\vec{r}_3+\vec{r}_2)-\dfrac{5c\sqrt{3-\overline{c}}}{14\sqrt{2}}\vec{r}_1, \nonumber\\
	\vec{q}_5 &= \dfrac{\sqrt{49\gamma^2+9800c\gamma\sqrt{3-\overline{c}} + \omega } }{420}\vec{r}_1, \label{qr}
\end{align}
where $\overline{c}=\cos(4\theta)$, $\gamma=\sqrt{ (75+25\overline{c})\sqrt{6-2\overline{c}}-72}$ and $\omega= 18125\cos(8\theta)-72500\cos(4\theta)-108706$. Finally, if we identify $\alpha=\dfrac{4\sqrt{2}}{\sqrt{3-\overline{c}}}$, then we recover the identity in (\ref{sos}).

\subsection{From SOS to self testing}
\label{SOS2ST}
In (\ref{append:sos}), we have obtained the SOS expression for the Bell inequality in (\ref{bell2}). Suppose that we obtain the maximum violation, $\bracket{\psi'}{\overline{\BB}(\alpha)}{\psi'}=\sum_\lambda \bracket{\psi'}{P_\lambda^\dag P_\lambda}{\psi'} = 0$. From the relations between $\vec{r}_\lambda$ and $\vec{q}_\lambda$ in (\ref{qr}), it must then be the case that $\vec{r}_\lambda\cdot\vec{V}\ket{\psi'}=0$, for all $\lambda$. 

To show that these relations can be used for self testing, it is sufficient to show the existence of dichotomic operators $X_A,Z_A,X_B$ and $Z_B$ such that the following relations in main text (Equations (4) main text) hold. For clarity, we reproduce the relations here.
\begin{align}
	Z_A\ket{\psi}&=Z_B\ket{\psi}, \label{cond1}\\
	\sin\theta X_A(I+Z_B)\ket{\psi}&=\cos\theta X_B (I-Z_A)\ket{\psi}, \label{cond3}
\end{align}

\noindent Let us define the operators as 

\begin{align}
	Z_A&=A_0, \hspace{1cm} &X_A &= A_1, \nonumber\\
	Z_B &=\dfrac{B_0+B_1}{2\cos\mu},&X_B &= \dfrac{B_0-B_1}{2\sin\mu}, \label{defineoperators}
\end{align}

\noindent where $\tan\mu=\sin(2\theta)$. The last two operators trivially satisfy the relation:

\begin{eqnarray}
&Z_B^2\cos^2(\mu)+X_B^2\sin^2(\mu)=\II,\label{extra_cond}
\end{eqnarray}

\noindent Also, it can be easily verified that eqs. (\ref{cond1}-\ref{cond3}) are equivalent to the following the identities:
\begin{align}
	\vec{r}_1\cdot\vec{V}\ket{\psi'}=0 , \nonumber\\
	\left[\dfrac{\sin\theta}{2\cos\mu} (\vec{r}_4+\vec{r}_5) + \dfrac{\cos\theta}{2\sin\mu} (\vec{r}_2-\vec{r}_3+\vec{r}_1) \right]\cdot\vec{V} \ket{\psi'}=0 . \label{sostransform}
\end{align}
However, it is still left to show that the operators defined in (\ref{defineoperators}) are dichotomic (i.e., unitary and hermitian). $Z_A$ and $X_A$ are clearly the case. From $\vec{r}_1\cdot\vec{V}\ket{\psi'}=0$, we have $(B_0+B_1)\ket{\psi'}=2\cos\mu A_0\ket{\psi'}$. Since $A_0$ is unitary and hermitian, it is necessary that $Z_B\equiv (B_0+B_1)/(2\cos\mu)$ is also unitary and hermitian. This can be seen easily by expanding the state $\ket{\psi'}$ locally in terms of Schmidt decomposition. From condition (\ref{extra_cond}), it then follows that $X_B$ is also unitary and hermitian.

To find out the state $\ket{\psi}=\cos\theta\ket{00}+\sin\theta\ket{11}$ that the correlations self test, we recall that $\alpha=\dfrac{4\sqrt{2}}{\sqrt{3-\cos(4\theta)}}$. Inverting the equation, we find that $\tan\theta=\sqrt{\frac{4-\alpha^2}{2\alpha^2}}$.

Thus, we have shown that if a state violates the $\BB(\alpha)$ maximally, then the state is equivalent, up to isometry, to the state $\ket{\psi}=\cos\theta\ket{00}+\sin\theta\ket{11}$ where $\tan\theta=\sqrt{\frac{4-\alpha^2}{2\alpha^2}}$.

\subsection{Robustness of SOS self testing}
\label{robust}

In this section, we would like to address the situation when the Bell violation is not maximal but close to it, possibly due to some experimental errors. Suppose that, instead of maximum violation, we have $\bracket{\psi}{\BB(\alpha)}{\psi'}\geq b(\alpha)-\epsilon^2$, i.e., $\bracket{\psi'}{\overline{\BB}(\alpha)}{\psi'} \leq \epsilon^2$. Is our system still close to the self tested state $\ket{\psi}$? 

\subsubsection{Sufficient Conditions for Robustness of self testing}

Using the SOS expression, we obtain $\sum_\lambda \bra{\psi'}P^\dagger_{\lambda} P_\lambda \ket{\psi'} \leq \epsilon^2$. We can hence derive an upper bound of the norm, $\bra{\psi'}P_\lambda^\dag P_\lambda \ket{\psi'} \leq \epsilon^2$, and so we have that
\begin{align}
	\left|\left|P_\lambda\ket{\psi'}\right|\right| \leq \epsilon. \label{bound}
\end{align}

Using (\ref{sostransform}), we obtain the following error terms
\begin{align}
	\norm{(Z_A-Z_B')\ket{\psi'}} \leq \epsilon, \\
	\norm{\Big(  \sin\theta X_A(I+Z_B')-\cos\theta X_B' (I-Z_A)        \Big)\ket{\psi'}} \nonumber\\
	\leq \left( \left|\dfrac{\sin\theta}{\cos\mu}\right| + \left|\dfrac{3\cos\theta}{2\sin\mu}\right| \right)\epsilon \equiv \epsilon_2 , 
\end{align}
where $Z_A,X_A,Z_B'$ and $X_B'$ have the same definitions as in (\ref{defineoperators}). However, at this point of time the operators $X_B'$ and $Z_B'$ are no longer unitary (although they are hermitian and satisfy eq. (\ref{extra_cond})), thus the prime notations. 
The above conditions are surprisingly sufficient for a robust self testing, as we will show below.

\subsubsection{Proof}
First of all, we would like to find unitary and hermitian operators $X_B$ and $Z_B$ such that they behave almost in a similar way to $X_B'$ and $Z_B'$ respectively. For $Z_B'$, let us define $Z_B\equiv Z_B'/|Z_B'|$, where the subspaces of $Z_B'$ with zero eigenvalues are defined to have eigenvalue 1 for the new operator. This new operator is both hermitian and unitary. Furthermore, it has the property that it acts in a very similar way to $Z_B'$ in the sense
\begin{align*}
	&\hspace{0.5cm}\norm{(Z_A-Z_B)\ket{\psi'}}= \\
	&=\norm{(\II-Z_AZ_B'/|Z_B'|)\ket{\psi'}}= \\
	&=\norm{(\II-Z_AZ_B'+Z_AZ_B'-Z_AZ_B'/|Z_B'|)\ket{\psi'}}\leq \\
	&\leq \norm{(\II-Z_AZ_B')\ket{\psi'}} + \norm{(Z_AZ_B'-Z_AZ_B'/|Z_B'|)\ket{\psi'}}\leq \\
	&\leq \epsilon + \sqrt{\bracket{\psi'}{(Z_AZ_B')^2+1-2|Z_AZ_B'|}{\psi'}}\leq \\
	&\leq \epsilon + \sqrt{\bracket{\psi'}{(Z_AZ_B')^2+1-2Z_AZ_B'}{\psi'}}= \\
	&= \epsilon + \norm{(\II-Z_AZ_B')\ket{\psi'}} \leq 2\epsilon.
\end{align*}
Analogously, from condition (\ref{extra_cond}) we have that $(Z_B')^2\leq \frac{\II}{\cos^2(\mu)}$. This, together with eq. (\ref{cond1}), implies that

\be
\|\II-(Z_B')^2\ket{\psi'}\|\leq \epsilon\left(\frac{1}{\cos(\mu)}+1\right). 
\label{Z_B_unit}
\ee

\noindent Using the relation $(Z_B-Z_B')^2=(1-|Z_B'|)^2\leq (1-|Z_B'|)^2(1+|Z'_B|)^2=[1-(Z_B')^2]^2$ we thus have that 

\be
\|(Z_B-Z_B')\ket{\psi'}\|\leq \epsilon\left(\frac{1}{\cos(\mu)}+1\right)\equiv \epsilon_3.
\ee

We now need an operator $X_B$ such that it is hermitian, unitary and behaves almost similar to $X_B'$. Consider $\norm{(X_B'^2-\II)\ket{\psi'}}$. By eq. (\ref{Z_B_unit}) and condition (\ref{extra_cond}) we have that
 
\begin{align*}
	\norm{(\II-X_B'^2)\ket{\psi'}} &= \norm{\left(\II-\frac{\II-\cos^2(\mu)(Z_B')^2}{\sin^2(\mu)}\right)\ket{\psi'}}\leq \\
	&\frac{\epsilon_3}{\tan^2(\mu)}.
\end{align*}

Define then $X_B\equiv X_B'/|X_B'|$. By the same argument we used to bound $\|(Z_B-Z_B')\ket{\psi}\|$, we arrive at

\be
\norm{(X_B-X_B')\ket{\psi'}}\leq \frac{\epsilon_3}{\tan^2(\mu)}\equiv \epsilon_4. 
\ee

Putting all together, we have that

\begin{align}
	\norm{(Z_A-Z_B)\ket{\psi'}} \leq 2\epsilon, \label{1}\\
	\norm{\Big(  \sin\theta X_A(\II+Z_B)-\cos\theta X_B(\II-Z_A) \Big)\ket{\psi'}} & \leq \epsilon'_2, \label{3}
\end{align}
where $\epsilon'_2=\epsilon_2+|\sin\theta|\epsilon_3+2|\cos\theta|\epsilon_4$. 

Let us see how one can obtain a robust self testing statement from (\ref{1}-\ref{3}). Recall the action of the isometry $\Phi$
\begin{align}
	 \Phi(\ket{\psi^{\prime}}) & = 
	 \frac{1}{4} (\II + Z_{A})  (\II + Z_{B}) \ket{\psi'}\ket{00} \nonumber\\
	 & +  \frac{1}{4} X_{B}(\II + Z_{A})   (\II - Z_{B}) \ket{\psi^{\prime}}\ket{01} \nonumber\\
	 & +  \frac{1}{4}X_{A}(\II- Z_{A})  (\II + Z_{B}) \ket{\psi^{\prime}}\ket{10}\nonumber\\
	 & +  \frac{1}{4} X_{A}   X_{B} (\II - Z_{A})(\II-Z_{B}) \ket{\psi^{\prime}}\ket{11}.\label{isometry}
\end{align}
We want to make some statements about this state compared to the ideal self tested state $\ket{\psi}=\cos\theta\ket{00}+\sin\theta\ket{11}$. The first term in (\ref{isometry}) can be approximated by $\frac{\II+Z_A}{2}\ket{\psi'}\ket{00}$, since

\be
\norm{ \frac{(\II + Z_{A})   (\II + Z_{B})}{4} \ket{\psi^{\prime}}\ket{00}-\frac{(\II + Z_{A})}{2}\ket{\psi^{\prime}}\ket{00}} \leq \epsilon.
\ee

\noindent The second term is close to zero, for

\begin{align}
	\frac{1}{4}\norm{ X_{B}(\II + Z_{A})   (\II - Z_{B}) \ket{\psi^{\prime}}\ket{01}} \leq \epsilon.
\end{align}
Similarly, the norm of the third term is bounded by the same amount. As for the last term in (\ref{isometry}), using (\ref{1}) and (\ref{3}), we can approximate it by $\tan(\theta)\frac{(\II+Z_A)}{2}\ket{\psi'}\ket{11}$ with an error of $\epsilon+ \frac{\epsilon_2'}{2|\cos(\theta)|}$. Thus we finally have 
\begin{align}
	&\norm{\Phi(\ket{\psi'})-\ket{junk}\ket{\psi}} \leq \overline{\epsilon} \nonumber\\
	&\hspace{2cm}+ \norm{ \left(\dfrac{\II+Z_A}{2\cos(\theta)}\ket{\psi'}-\ket{junk} \right)\ket{\psi} },
\end{align}

\noindent with $\overline{\epsilon}=4\epsilon + \frac{\epsilon_2'}{2|\cos(\theta)|}$.

We would like to identify 
\begin{align}
	\ket{junk}=\dfrac{(\II+Z_A)}{2\cos(\theta)}\ket{\psi'},
\end{align}
but the state on the right hand side may not be normalized. To do this, we have to bound the norm of the state $\frac{\II+Z_A}{2}\ket{\psi'}$. This can be done by noticing that the isometry $\Phi$ preserves the norm. By considering $\norm{\Phi(\ket{\psi'})}=1$, we obtain
\begin{align}
	1-\overline{\epsilon}\leq\norm{\dfrac{(\II+Z_A)}{2c}\ket{\psi'}} \leq 1 + \overline{\epsilon}.
\end{align}
Considering this uncertainty in the norm of the state $\ket{junk}$ thus we have the final robust bound 
\begin{align}
	\norm{\Phi(\ket{\psi'})-\ket{junk}\ket{\psi}} \leq  2\overline{\epsilon}.
\end{align}
This completes the proof of robust self testing on states with correlation close to the maximum violation of the Bell inequality $\BB(\alpha)$.

\subsection{General self testing for Any Dimension \label{selftestingd}}
\label{general}

\subsubsection{Isometry for High Dimension and Sufficient Conditions for Self Testing}

In the main text, in Figure 2, we show the circuit for high dimensional self testing. We will now provide the sufficient conditions for self testing. Let $\{P^{(i)}_A\}_{i=0}^{d-1}$ ($\{P^{(i)}_B\}_{i=0}^{d-1}$) be a complete set of orthogonal projectors in Alice's (Bob's) Hilbert space. If we have a set of $d-1$ local unitary operators, $X_A^{(k)}$ and $X_B^{(k)}$, where $k=0,1,\ldots,d-1$, satisfying
\begin{align}
	Z_{A/B} &= \sum_{i=0}^{d-1} \omega^i P_{A/B}^{(i)}, \label{cond11}\\
	P_A^{(i)}\ket{\psi'}&=P_B^{(i)}\ket{\psi'}, \;\;\forall i, \label{cond22} \\
	X_{A}^{(i)}P_{B}^{(i)}\ket{\psi'}&=\tan\delta_i (X_{B}^{(i)})^\dag P_{A}^0 \ket{\psi'}, \;\;\forall i  \label{cond33}
\end{align}
then the circuit in Figure 2 can be used for self testing. Recall that the control phase gate, $S_{A/B}$ and the control rotation gate, $R_{A/B}$ in the circuit are defined via
\begin{align}
	S_{A/B}\ket{k}\ket{\psi'} \rightarrow \ket{k}Z_{A/B}^k\ket{\psi'}, \\
	R_{A/B}\ket{k}\ket{\psi'} \rightarrow \ket{k}X_{A/B}^{(k)}\ket{\psi'},
\end{align}
One can easily check that if we input the state $\ket{\psi'}\ket{00}$ into the circuit we obtain the state 
\begin{align}
	\Phi(\ket{\psi'}) &= \ket{junk} \dfrac{\ket{00}+\sum_{i=1}^{d-1} \tan\delta_i\ket{ii}}{\sqrt{1+\sum_{i=1}^{d-1}\tan\delta_i}}, \\
	&= \ket{junk}\ket{\psi}
\end{align}
as output, where we have managed to self test the system $\ket{\psi'}\widetilde{=}\ket{\psi}$. Thus eqs. (\ref{cond11}-\ref{cond33}) are sufficient conditions for self testing.

\subsubsection{Correlations for Maximally Entangled Self Testing} \label{correlations}

Here we will provide the correlations sufficient to self test a $d$ dimensional maximally entangled states
\begin{align}
	\ket{\phi}=\dfrac{1}{\sqrt{d}} \sum_{i=0}^{d-1} \ket{ii}.
\end{align}
We will only need 3 measurements each on Alice's and Bob's side labelled by $(A_0,A_1,A_2)$ and $(B_0,B_1,B_2)$. Firstly, consider the correlations between the outcomes of $A_0$ and $B_0$ as shown in Table (\ref{a0b0}).
\begin{table}[htbp!]
	\centering

\begin{tabular}{|cc|c|c|c|c|c|c|c|c|c|c|}
\hline
& & \multicolumn{6}{c|}{$B_0$}  \\ \cline{3-8}
& & $\pb{0}{0}$ & $\pb{0}{1}$ & $\pb{0}{2}$ & $\pb{0}{3}$ & $\cdots$ & $\pb{0}{d-1}$ \\ \cline{1-8}
\multicolumn{1}{|c}{\multirow{6}{*}{$A_0$}} &
\multicolumn{1}{|c|}{$\pa{0}{0}$} & $\frac{1}{d}$ & 0 & 0 & 0& $\cdots$ &0    \\ \cline{2-8}
\multicolumn{1}{|c}{}                        &
\multicolumn{1}{|c|}{$\pa{0}{1}$} & 0 & $\frac{1}{d}$ & 0 & 0& $\cdots$ &0   \\ \cline{2-8}
\multicolumn{1}{|c}{}                        &
\multicolumn{1}{|c|}{$\pa{0}{2}$} & 0 & 0 & $\frac{1}{d}$ & 0&$\cdots$   &0 \\ \cline{2-8}
\multicolumn{1}{|c}{}         &
\multicolumn{1}{|c|}{$\pa{0}{3}$} & 0 & 0 & 0& $\frac{1}{d}$ &$\cdots$   &0 \\ \cline{2-8}
\multicolumn{1}{|c}{}         &
\multicolumn{1}{|c|}{$\vdots$} & $\vdots$ & $\vdots$ & $\vdots$ & $\vdots$ & $\ddots$  &$\vdots$  \\ \cline{2-8}
\multicolumn{1}{|c}{}         &
\multicolumn{1}{|c|}{$\pa{0}{d-1}$} & 0 & 0 & 0 & 0& $\cdots$ &$\frac{1}{d}$    \\ \cline{1-8}
\end{tabular}
	\caption{Correlations for $A_0$ and $B_0$. \label{a0b0}}
\end{table}
From Table (\ref{a0b0}), the correlations can be summarized as $\bracket{\psi'}{\pa{0}{i}\pb{0}{j}}{\psi'}=\delta_{i,j}/d$. Since the outcomes are represented by projective operators, we can deduce that
\begin{align}
	\pa{0}{i}\ket{\psi'}=\pb{0}{i}\ket{\psi'}.
\end{align}
To proceed, we need the correlations between $A_0$ and $B_1$, as shown in Table (\ref{a0b1}). We assume that $d$ is even, the case when $d$ is odd can be generalized accordingly.

	\begin{table}[htbp!]
		\centering
\begin{tabular}{|cc|c|c|c|c|c|c|c|c|c|c|c|}
\hline
& & \multicolumn{7}{c|}{$B_1$}  \\ \cline{3-9}
& & $\pb{1}{0}$ & $\pb{1}{1}$ & $\pb{1}{2}$ & $\pb{1}{3}$ & $\cdots$ & $\pb{1}{d-2}$ & $\pb{1}{d-1}$ \\ \cline{1-9}
\multicolumn{1}{|c}{\multirow{7}{*}{$A_0$}} &
\multicolumn{1}{|c|}{$\pa{0}{0}$} & $\frac{1}{d}\cos^2\theta_{0}$ & $\frac{1}{d}\sin^2\theta_{0}$ & 0 & 0& $\cdots$ &0 & 0     \\ \cline{2-9}
\multicolumn{1}{|c}{}                        &
\multicolumn{1}{|c|}{$\pa{0}{1}$} & $\frac{1}{d}\sin^2\theta_{0}$ & $\frac{1}{d}\cos^2\theta_{0}$ & 0 & 0& $\cdots$ &0 &0  \\ \cline{2-9}
\multicolumn{1}{|c}{}                        &
\multicolumn{1}{|c|}{$\pa{0}{2}$} & 0 & 0 & $\frac{1}{d}\cos^2\theta_{1}$ & $\frac{1}{d}\sin^2\theta_{1}$ &$\cdots$   &0 &0\\ \cline{2-9}
\multicolumn{1}{|c}{}         &
\multicolumn{1}{|c|}{$\pa{0}{3}$} & 0 & 0 & $\frac{1}{d}\sin^2\theta_{1}$& $\frac{1}{d}\cos^2\theta_{1}$ & $\cdots$ &0 &0 \\ \cline{2-9}
\multicolumn{1}{|c}{}         &
\multicolumn{1}{|c|}{$\vdots$} & $\vdots$ & $\vdots$ & $\vdots$ & $\vdots$ & $\ddots$  &$\vdots$ &$\vdots$ \\ \cline{2-9}
\multicolumn{1}{|c}{}         &
\multicolumn{1}{|c|}{$\pa{0}{d-2}$} & 0 & 0 &0&0&$\cdots$& $\frac{1}{d}\cos^2\theta_{\frac{d}{2}-1}$& $\frac{1}{d}\cos^2\theta_{\frac{d}{2}-1}$  \\ \cline{2-9}
\multicolumn{1}{|c}{}         &
\multicolumn{1}{|c|}{$\pa{0}{d-1}$} & 0 & 0 & 0 & 0& $\cdots$ & $\frac{1}{d}\sin^2\theta_{\frac{d}{2}-1}$ & $\frac{1}{d}\cos^2\theta_{\frac{d}{2}-1}$   \\ \cline{1-9}
\end{tabular}
	\caption{Correlations for $A_0$ and $B_1$. \label{a0b1}}
\end{table}

In fact, the correlations for the pairs $(A_1,B_0)$ and $(A_1,B_1)$ have the same structure: the outcomes of one measurement on Alice's side only have non zero correlation with the corresponding projectors in the same block. Thus for the moment we shall focus on one particular block of the correlations, with the outcomes $(\pa{0}{2m},\pa{0}{2m+1},\pa{1}{2m},\pa{1}{2m+1},\pb{0}{2m},\pb{0}{2m+1},\pb{1}{2m},\pb{1}{2m+1})$, where $m=0,1,\ldots,\frac{d}{2}-1$. For these projectors, we have the correlations summarized as shown in Table (\ref{a0b0block},\ref{a0b1block},\ref{a1b0block},\ref{a1b1block}). 
\begin{table}[htbp!]
	\centering
\begin{tabular}{|cc|c|c|}
	\hline
& & \multicolumn{2}{c|}{$B_0$}  \\ \cline{3-4}
& & $\pb{0}{2m}$ & $\pb{0}{2m+1}$  \\ \cline{1-4}
\multicolumn{1}{|c}{\multirow{2}{*}{$A_0$}} &
\multicolumn{1}{|c|}{$\pa{0}{2m}$} & $\frac{1}{d}$ &  0      \\ \cline{2-4}
\multicolumn{1}{|c}{}                        &
\multicolumn{1}{|c|}{$\pa{0}{2m+1}$} & 0 & $\frac{1}{d}$       \\ \cline{1-4}
\end{tabular}
\caption{Correlations for $A_0$ and $B_0$.\label{a0b0block}}
\end{table}
\begin{table}[htbp!]
	\centering
\begin{tabular}{|cc|c|c|}
	\hline
& & \multicolumn{2}{c|}{$B_1$}  \\ \cline{3-4}
& & $\pb{0}{2m}$ & $\pb{0}{2m+1}$  \\ \cline{1-4}
\multicolumn{1}{|c}{\multirow{2}{*}{$A_0$}} &
\multicolumn{1}{|c|}{$\pa{0}{2m}$} & $\frac{1}{d}\cos^2\theta_{m}$ &  $\frac{1}{d}\sin^2\theta_{m}$      \\ \cline{2-4}
\multicolumn{1}{|c}{}                        &
\multicolumn{1}{|c|}{$\pa{0}{2m+1}$} & $\frac{1}{d}\sin^2\theta_{m}$ &  0$\frac{1}{d}\cos^2\theta_{m}$      \\ \cline{1-4}
\end{tabular}
\caption{Correlations for $A_0$ and $B_1$.\label{a0b1block}}
\end{table}
\begin{table}[htbp!]
	\centering
\begin{tabular}{|cc|c|c|}
	\hline
& & \multicolumn{2}{c|}{$B_0$}  \\ \cline{3-4}
& & $\pb{0}{2m}$ & $\pb{0}{2m+1}$  \\ \cline{1-4}
\multicolumn{1}{|c}{\multirow{2}{*}{$A_1$}} &
\multicolumn{1}{|c|}{$\pa{0}{2m}$} & $\frac{1}{d}\cos^2\phi_{m}$ &  $\frac{1}{d}\sin^2\phi_{m}$      \\ \cline{2-4}
\multicolumn{1}{|c}{}                        &
\multicolumn{1}{|c|}{$\pa{0}{2m+1}$} & $\frac{1}{d}\sin^2\phi_{m}$ &  0$\frac{1}{d}\cos^2\phi_{m}$      \\ \cline{1-4}
\end{tabular}
\caption{Correlations for $A_1$ and $B_0$.\label{a1b0block}}
\end{table}
\begin{table}[htbp!]
	\centering
\begin{tabular}{|cc|c|c|}
	\hline
& & \multicolumn{2}{c|}{$B_1$}  \\ \cline{3-4}
& & $\pb{1}{2m}$ & $\pb{1}{2m+1}$  \\ \cline{1-4}
\multicolumn{1}{|c}{\multirow{2}{*}{$A_1$}} &
\multicolumn{1}{|c|}{$\pa{1}{2m}$} & $\frac{1}{d}\cos^2(\theta_{m}-\phi_m)$ &  $\frac{1}{d}\sin^2(\theta_m-\phi_{m})$      \\ \cline{2-4}
\multicolumn{1}{|c}{}                        &
\multicolumn{1}{|c|}{$\pa{1}{2m+1}$} & $\frac{1}{d}\sin^2(\theta_m-\phi_{m})$ &  $\frac{1}{d}\cos^2(\theta_{m}-\phi_m)$      \\ \cline{1-4}
\end{tabular}
\caption{Correlations for $A_1$ and $B_1$.\label{a1b1block}}
\end{table}

The correlations are non zero only in a block diagonal manner. Within each block, the correlations are labelled by the parameter $\theta_m,\phi_m$, which can always be chosen to be $0\leq\theta_m,\phi_m\leq \pi/2$. First of all, recall that
\begin{align}
	\pa{0}{i}\ket{\psi'}=\pb{0}{i}\ket{\psi'}, \label{proof1}
\end{align}
and we define the operators
\begin{align}
	O^{B_0}_m\equiv \pb{0}{2m}-\pb{0}{2m+1}, \\
	O^{B_1}_m\equiv \pb{1}{2m}-\pb{1}{2m+1}, \\
	O^{A_0}_m\equiv \pa{0}{2m}-\pb{0}{2m+1}, \\
	O^{A_1}_m\equiv \pa{1}{2m}-\pb{1}{2m+1}. 	\label{proof2}
\end{align}
From (\ref{proof1}) and (\ref{proof2}), it is easy to check that $\bracket{\psi'}{O^{B_0}_m(O^{B_1}_m-\cos\phi_m O^{B_0}_m)}{\psi'}=0$. Thus $O^{B_0}_m\ket{\psi'}\perp (O^{B_1}_m-\cos\phi_m O^{B_0}_m)\ket{\psi'}$. It can also be checked that the norm of the two vectors are given by
\begin{align}
	\norm{O^{B_0}_m\ket{\psi'}}&=\sqrt{\frac{2}{d}}, \label{proof30}\\
	\norm{(O^{B_1}_m-\cos\phi_m O^{B_0}_m)\ket{\psi'}}&=\sqrt{\frac{2}{d}}\sin\phi_m. \label{proof3}
\end{align}
Since the two vectors (\ref{proof30}) and (\ref{proof3}) are orthogonal, we can try to decompose the vector $O^{A_1}_m\ket{\psi'}$ into them. The span of the two orthogonal vectors (\ref{proof30}) and (\ref{proof3}) may not contain the vector $O^{A_1}_m\ket{\psi'}$ in general but as well see, in this case, it does. Since we have 
\begin{align}
	\bracket{\psi'}{O^{B_0}_m O^{A_1}_m}{\psi'}= \frac{2}{d}\cos\theta_m, \\
	\bracket{\psi'}{(O^{B_1}_m-\cos\phi_mO^{B_0}_m) O^{A_1}_m}{\psi'}= \frac{2}{d}\sin\theta_m,
\end{align}
we can decompose 
\begin{align}
	O^{A_1}_m\ket{\psi'} = \cos\theta_m O^{B_0}_m\ket{\psi'} + \sin\theta_m (O^{B_1}_m-\cos\phi_mO^{B_0}_m)\ket{\psi'}. \label{proof5}
\end{align}
Since, the norm on the left hand side equal the norm on the right hand side we conclude that the decomposition in (\ref{proof5}) is complete. Notice that $(O^{B_1}_m)^2=\II^{B_1}_m$, where $\II^{B_1}_m$ is the identity on the subspace defined by the outcomes $\pb{1}{2m},\pb{1}{2m+1}$. Similarly, we have 
\begin{align}
	(O^{B_0}_m)^2=\II^{B_0}_m, \nonumber\\
	(O^{B_1}_m)^2=\II^{B_1}_m, \nonumber\\
	(O^{A_0}_m)^2=\II^{A_0}_m, \nonumber\\
	(O^{A_1}_m)^2=\II^{A_1}_m. \label{proof6}
\end{align}
The 4 identity operators in (\ref{proof6}) may not be the same in general. However for the correlations shown in Table (\ref{a0b0block},\ref{a0b1block},\ref{a1b0block},\ref{a1b1block}), they satisfy the condition
\begin{align}
	\bracket{\psi'}{I^{B_i}_mI^{A_j}_m}{\psi'}=\norm{I^{B_i}_m\ket{\psi'}}\norm{I^{A_j}_m\ket{\psi'}}. \label{proof9}
\end{align}
for all $i,j={0,1}$. Thus, from Cauchy-Schwartz inequality, we conclude that $I^{A_i}_m\ket{\psi'}=I^{B_j}_m\ket{\psi'}$ for all $i,j={0,1}$. Since their actions on the state $\ket{\psi'}$ are the same, we will simply write $\II_m$ as opposed to the notations in (\ref{proof6}), but bearing in mind that it still depends on $m$, which refers to the different block.

Now, from (\ref{proof5}), we square both sides and after simplifying we obtain
\begin{align}
	\{O^{B_0}_m,O^{B_1}_m\}\ket{\psi'} &= 2\cos\phi_m \II_m\ket{\psi'}, \\
	\frac{\{O^{B_0}_m,O^{B_1}_m\}}{2\cos\phi_m} \ket{\psi'} &= \II_m\ket{\psi'}. \label{proof7}
\end{align}
Since $\ket{\psi'}$ is a bipartite state, by decomposing it into its Schmidt basis, we conclude that the two operators on both side in (\ref{proof7}) are equal, 
\begin{align}
	\frac{\{O^{B_0}_m,O^{B_1}_m\}}{2\cos\phi_m} = \II_m. \label{proof8}
\end{align}
We are now ready to define a unitary operators for our self testing,
\begin{align}
	X^B_m \equiv \dfrac{O^{B_1}_m-\cos\phi_m O^{B_0}_m}{\sin\phi_m} + (\II-\II_m), \label{proof10}
\end{align}
which are unitary and hermitian. It is easy to check that it is unitary by using (\ref{proof8}), or in other words, $(X^B_m)^2=I$. Another piece of important information is that the operator in (\ref{proof10}) satisfy the equation
\begin{align}
	X^B_m \pb{0}{2m+1}\ket{\psi'}=\pb{0}{2m}X^B_m\ket{\psi'}. \label{proof11}
\end{align}
which is very similar to (\ref{cond33}). Similarly, we can do the same analysis on Alice's side, and define
\begin{align}
	X^A_m \equiv \dfrac{O^{A_1}_m-\cos\theta_m O^{A_0}_m}{\sin\theta_m} + (\II-\II_m), \label{proof12}
\end{align}
which has the property
\begin{align}
	X^A_m \pa{0}{2m+1}\ket{\psi'}=\pa{0}{2m}X^A_m\ket{\psi'}. \label{proof13}
\end{align}
In addition, using the decomposition in (\ref{proof5}), it is also easy to show that 
\begin{align}
	X^A_m\ket{\psi'}=X^B_m\ket{\psi'}. \label{proof13new}
\end{align}

Thus from correlations of the four measurements $(A_0,A_1,B_0,B_1)$, we managed to construct operators which rotate all $\pa{0}{2m+1}\ket{\psi'} \rightarrow \pa{0}{2m}\ket{\psi'}$ and $\pb{0}{2m+1}\ket{\psi'}\rightarrow \pb{0}{2m}\ket{\psi'}$ for all $m=0,1,\ldots,\frac{d}{2}-1$. However, the goal is to be able to transform all  $\pa{0}{j}\ket{\psi'} \rightarrow \pa{0}{0}\ket{\psi'}$ and $\pb{0}{j}\ket{\psi'}\rightarrow \pb{0}{0}\ket{\psi'}$ for all $j$.

To accomplish this, we require a third set of measurements, $(A_2,B_2)$. In contrast to $(A_1,B_1)$ which has block correlations with $(A_0,B_0)$ for the outcomes $(0,1),(2,3),\ldots,(2m,2m+1),\ldots,(d-2,d-1)$, the measurements $(A_2,B_2)$ has block correlations with $(A_0,B_0)$ for the outcomes $(1,2),(3,4),\ldots,,(2n+1,2n+2),\ldots,(d-1,0)$, where $n$ runs from $0$ to $\frac{d}{2}-1$. Note that, for clarity, we shall refer to the dummy variable $m$ for the block correlations between $(A_0,A_1,B_0,B_1)$ and $n$ for block correlations between $(A_0,A_2,B_0,B_2)$. Also, the outcomes $\pa{0}{d}$ is understood to be $\pa{0}{0}$. 

For instance, the correlations between $A_0$ and $B_2$ has the form as shown in Table (\ref{a0b2}).

\begin{table}[htbp!]
	\centering
\begin{tabular}{|cc|c|c|c|c|c|c|c|c|c|c|c|}
\hline
& & \multicolumn{7}{c|}{$B_2$}  \\ \cline{3-9}
& & $\pb{2}{0}$ & $\pb{2}{1}$ & $\pb{2}{2}$ & $\pb{2}{3}$ & $\cdots$ & $\pb{2}{d-2}$ & $\pb{2}{d-1}$ \\ \cline{1-9}
\multicolumn{1}{|c}{\multirow{7}{*}{$A_0$}} &
\multicolumn{1}{|c|}{$\pa{0}{0}$} & $\frac{1}{d}\cos^2\theta'_{\frac{d}{2}-1}$ &0 & 0 & 0& $\cdots$ &0 &  $\frac{1}{d}\sin^2\theta'_{\frac{d}{2}-1}$     \\ \cline{2-9}
\multicolumn{1}{|c}{}                        &
\multicolumn{1}{|c|}{$\pa{0}{1}$} & 0 & $\frac{1}{d}\cos^2\theta'_{0}$ & $\frac{1}{d}\sin^2\theta'_{0}$  & 0& $\cdots$ &0 &0  \\ \cline{2-9}
\multicolumn{1}{|c}{}                        &
\multicolumn{1}{|c|}{$\pa{0}{2}$} & 0 & $\frac{1}{d}\sin^2\theta'_{0}$ & $\frac{1}{d}\cos^2\theta'_{0}$ &0 &$\cdots$   &0 &0\\ \cline{2-9}
\multicolumn{1}{|c}{}         &
\multicolumn{1}{|c|}{$\pa{0}{3}$} & 0 & 0 & 0 & $\frac{1}{d}\cos^2\theta'_{1}$ & $\cdots$ &0 &0 \\ \cline{2-9}
\multicolumn{1}{|c}{}         &
\multicolumn{1}{|c|}{$\vdots$} & $\vdots$ & $\vdots$ & $\vdots$ & $\vdots$ & $\ddots$  &$\vdots$ &$\vdots$ \\ \cline{2-9}
\multicolumn{1}{|c}{}         &
\multicolumn{1}{|c|}{$\pa{0}{d-1}$} & 0 & 0 & 0 & 0& $\cdots$ & $\frac{1}{d}\sin^2\theta_{\frac{d}{2}-2}$ & 0  \\ \cline{2-9}
\multicolumn{1}{|c}{}         &
\multicolumn{1}{|c|}{$\pa{0}{d-2}$} & $\frac{1}{d}\sin^2\theta'_{\frac{d}{2}-1}$ & 0 &0&0&$\cdots$ & 0 & $\frac{1}{d}\cos^2\theta'_{\frac{d}{2}-1}$  \\ \cline{1-9}
\end{tabular}
	\caption{Correlations for $A_0$ and $B_2$. \label{a0b2}}
\end{table}

In this case, the block correlations are given by
\begin{table}[htbp!]
	\centering
\begin{tabular}{|cc|c|c|}
	\hline
& & \multicolumn{2}{c|}{$B_0$}  \\ \cline{3-4}
& & $\pb{0}{2m}$ & $\pb{0}{2m+1}$  \\ \cline{1-4}
\multicolumn{1}{|c}{\multirow{2}{*}{$A_0$}} &
\multicolumn{1}{|c|}{$\pa{0}{2m}$} & $\frac{1}{d}$ &  0      \\ \cline{2-4}
\multicolumn{1}{|c}{}                        &
\multicolumn{1}{|c|}{$\pa{0}{2m+1}$} & 0 & $\frac{1}{d}$       \\ \cline{1-4}
\end{tabular}
\caption{Correlations for $A_0$ and $B_0$.\label{a0b0blocknew}}
\end{table}
\begin{table}[htbp!]
	\centering
\begin{tabular}{|cc|c|c|}
	\hline
& & \multicolumn{2}{c|}{$B_2$}  \\ \cline{3-4}
& & $\pb{2}{2m}$ & $\pb{2}{2m+1}$  \\ \cline{1-4}
\multicolumn{1}{|c}{\multirow{2}{*}{$A_0$}} &
\multicolumn{1}{|c|}{$\pa{0}{2m}$} & $\frac{1}{d}\cos^2\theta'_{m}$ &  $\frac{1}{d}\sin^2\theta'_{m}$      \\ \cline{2-4}
\multicolumn{1}{|c}{}                        &
\multicolumn{1}{|c|}{$\pa{0}{2m+1}$} & $\frac{1}{d}\sin^2\theta'_{m}$ &  0$\frac{1}{d}\cos^2\theta'_{m}$      \\ \cline{1-4}
\end{tabular}
\caption{Correlations for $A_0$ and $B_2$.\label{a0b2block}}
\end{table}
\begin{table}[htbp!]
	\centering
\begin{tabular}{|cc|c|c|}
	\hline
& & \multicolumn{2}{c|}{$B_0$}  \\ \cline{3-4}
& & $\pb{0}{2m}$ & $\pb{0}{2m+1}$  \\ \cline{1-4}
\multicolumn{1}{|c}{\multirow{2}{*}{$A_2$}} &
\multicolumn{1}{|c|}{$\pa{2}{2m}$} & $\frac{1}{d}\cos^2\phi'_{m}$ &  $\frac{1}{d}\sin^2\phi'_{m}$      \\ \cline{2-4}
\multicolumn{1}{|c}{}                        &
\multicolumn{1}{|c|}{$\pa{2}{2m+1}$} & $\frac{1}{d}\sin^2\phi'_{m}$ &  0$\frac{1}{d}\cos^2\phi'_{m}$      \\ \cline{1-4}
\end{tabular}
\caption{Correlations for $A_2$ and $B_0$.\label{a2b0block}}
\end{table}
\begin{table}[htbp!]
	\centering
\begin{tabular}{|cc|c|c|}
	\hline
& & \multicolumn{2}{c|}{$B_2$}  \\ \cline{3-4}
& & $\pb{2}{2m}$ & $\pb{2}{2m+1}$  \\ \cline{1-4}
\multicolumn{1}{|c}{\multirow{2}{*}{$A_2$}} &
\multicolumn{1}{|c|}{$\pa{2}{2m}$} & $\frac{1}{d}\cos^2(\theta'_{m}-\phi'_m)$ &  $\frac{1}{d}\sin^2(\theta'_m-\phi'_{m})$      \\ \cline{2-4}
\multicolumn{1}{|c}{}                        &
\multicolumn{1}{|c|}{$\pa{2}{2m+1}$} & $\frac{1}{d}\sin^2(\theta'_m-\phi'_{m})$ &  $\frac{1}{d}\cos^2(\theta'_{m}-\phi'_m)$      \\ \cline{1-4}
\end{tabular}
\caption{Correlations for $A_2$ and $B_2$.\label{a2b2block}}
\end{table}
By doing the same analysis as above, we will obtain two additional local operators on Alice's and Bob's side defined by
\begin{align}
	Y^A_n &\equiv \dfrac{N^{A_1}_n-\cos\theta'_n N^{A_0}_n}{\sin\theta'_n} + (\II-\II_n), \label{proof15} \\
	Y^B_n &\equiv \dfrac{N^{B_1}_n-\cos\phi'_n N^{B_0}_n}{\sin\phi'_n} + (\II-\II_n), \label{proof16}
\end{align}
where
\begin{align*}
	N^{b_0}_m\equiv \pb{0}{2m}-\pb{0}{2m+1}, \\
	N^{b_1}_m\equiv \pb{2}{2m}-\pb{2}{2m+1}, \\
	N^{a_0}_m\equiv \pa{0}{2m}-\pb{0}{2m+1}, \\
	N^{a_1}_m\equiv \pa{2}{2m}-\pb{2}{2m+1}.
\end{align*}
They have the property
\begin{align}
	Y^A_n \pa{0}{2n+2}\ket{\psi'}=\pa{0}{2n+1}Y^A_n\ket{\psi'}, \label{proof17} \\
	Y^B_n \pb{0}{2n+2}\ket{\psi'}=\pb{0}{2n+1}Y^B_n\ket{\psi'}. \label{proof18}
\end{align}
We now have local unitary operators such that we can transform all the vectors $\pa{0}{2n+2}\ket{\psi'}\rightarrow \pa{0}{2n+1}\ket{\psi'}$ and $\pb{0}{2n+2}\ket{\psi'}\rightarrow \pb{0}{2n+1}\ket{\psi'}$. Furthermore, they satisfy the relation
\begin{align}
	Y^A_n\ket{\psi'}=Y^B_n\ket{\psi'}. \label{proof21}
\end{align}

We are now ready to construct the sufficient conditions in (\ref{cond11}-\ref{cond33}) to complete the self testing. For (\ref{cond11}), it is simply a definition and there is no need to proof it. For (\ref{cond22}), we can easily take the projectors for the outcomes from the measurements $A_0$ and $B_0$, since all their projectors satisfy (\ref{proof1}), by taking $P_A^{(i)}=\pa{0}{i}$ and $P_B^{(i)}=\pb{0}{i}$ we thus have the first condition (\ref{cond22}). For third condition, we need a suitable local unitary operator on Alice's and Bob's side. From (\ref{proof13}) and (\ref{proof17}), it is easy to see that to achieve this we just need to define
\begin{align}
	X_A^{(i)} \equiv \left\{  \begin{array}{cl}
		X^A_0Y^A_0X^A_1Y^A_1\ldots Y^A_{k-1}X^A_k,    &\textrm{for odd }i=2k+1,  \\
		X^A_0Y^A_0X^A_1Y^A_1\ldots X^A_{k-1}Y^A_{k-1},    &\textrm{for even }i=2k,
	\end{array}    \right.. \\
	X_B^{(i)} \equiv \left\{  \begin{array}{cl}
		X^B_0Y^B_0X^B_1Y^B_1\ldots Y^B_{k-1}X^B_k,    &\textrm{for odd }i=2k+1,  \\
		X^B_0Y^B_0X^B_1Y^B_1\ldots X^B_{k-1}Y^B_{k-1},    &\textrm{for even }i=2k,
	\end{array}    \right.,
\end{align}
Indeed, then we have for the case of even $i=2k$,
\begin{align*}
	 &X_A^{(i)}P_B^{(i)} \ket{\psi'} \\
	=& X_A^{(i)}P_A^{(i)}\ket{\psi'} , \\
	=& X^A_0Y^A_0X^A_1Y^A_1\ldots X^A_{k-1}Y^A_{k-1} \pa{0}{2k}\ket{\psi'}, \\
	=& X^A_0Y^A_0X^A_1Y^A_1\ldots X^A_{k-1} \pa{0}{2(k-1)+1} Y^A_{k-1}\ket{\psi'}, \\
	=& X^A_0Y^A_0X^A_1Y^A_1\ldots  \pa{0}{2(k-1)} X^A_{k-1} Y^B_{k-1}\ket{\psi'}, \\
	=& \pa{0}{0} Y^B_{k-1}X^B_{k-1}\ldots Y^B_1X^B_1Y^B_0X^B_0 \ket{\psi'}, \\
	=& (X_B^{(i)})^\dag P_A^0 \ket{\psi'}, \;\;\;\forall i
\end{align*}
which is simply the condition (\ref{cond33}) with $\tan\delta_i=1$ for all $i$. The situation when $i$ is odd is similar. We thus complete the proof, and the state $\ket{\psi'}$ which produces the correlations above are self tested into maximally entangled states $\ket{\phi}=\sum_{i=0}^{d-1}\ket{ii}/\sqrt{d}$.

\end{document}